\documentclass[10pt,rqno,oneside]{article}

\usepackage{mathptmx} % CHANGES THE FONT - LETTERS ARE THICKER AND DARKER.
\usepackage{titlesec}
\titleformat*{\section}{\Large\bfseries} % INDICATES THE FONT SIZE OF THE SECTIONS. 
\titleformat*{\subsection}{\large\bfseries} % INDICATES THE FONT SIZE OF THE SUBSECTIONS. 

\usepackage{cite}
\usepackage[small]{caption}
\usepackage{amsmath}
\usepackage{amssymb}
\usepackage{amsfonts}
\usepackage{setspace}
\usepackage{latexsym}
\usepackage{mathrsfs}
\usepackage{epsfig}
\usepackage{amsthm}
\usepackage{yfonts}
\usepackage{graphicx}
\usepackage{color}
\usepackage{MnSymbol}
\usepackage{enumitem}
\setitemize{wide}

\numberwithin{equation}{section}

\newcommand{\n}{\noindent}
\newcommand{\be}{\begin{equation}}
\newcommand{\ee}{\end{equation}}
\newcommand{\ben}{\begin{displaymath}}
\newcommand{\een}{\end{displaymath}}
\newcommand{\vs}{\vspace{0.2cm}}

\usepackage{amsbsy}
\usepackage[margin=2cm]{caption} %CHANGES THE WIDTH OF THE CAPTION IN FIGURES.

\usepackage{fancyhdr}
\AtBeginDocument{\thispagestyle{plain}}
\fancypagestyle{plain}
	{\fancyhead{}
	\fancyfoot{}
	\fancyhead[LE,RO]{}
	\fancyhead[LO,RE]{}
	\fancyfoot[R]{\thepage}
	
	\headsep = 20pt}

\newtheorem{Theorem}{Theorem}[section]

\newtheorem{Remark}[Theorem]{Remark}

\newtheorem{Proposition}[Theorem]{Proposition}

\newtheorem{Aux-Lemma}[Theorem]{Aux-Lemma}
\newtheorem*{Acknowledgment}{Acknowledgments}

\newcommand{\dist}{{\rm dist}}

\newcommand{\An}{\mathcal{A}}
\newcommand{\stM}{{\bf M}}
\newcommand{\sM}{\Sigma}
\newcommand{\stg}{{\bf g}}
\newcommand{\sg}{g}
\newcommand{\stgeo}{\Gamma}
\newcommand{\sgeo}{\gamma}
\newcommand{\const}{\eta}
\newcommand{\length}{{\rm length}}

\begin{document}

%\n \begin{minipage}[l]{13.5cm}
%\begin{spacing}{1.5}
\n {\huge The asymptotic of static isolated systems and a generalised uniqueness for Schwarzschild.} 
%\end{spacing}
%\end{minipage}

\vs\vs

\n {\sc Martin Reiris}

\n {{\tt martin@aei.mpg.de}}

\vspace{.1cm}

\n {\sc Max Planck Institute f\"ur Gravitationsphysik}

\n {\tt Am M\"uhlenberg 1 D-14476 Golm, Germany}
\vs\vs

\hspace{.7cm}
\begin{minipage}[l]{11cm}
\begin{spacing}{.9}{\small It is proved that any static system that is spacetime-geodesically complete at infinity, and whose spacelike-topology outside a compact set is that of $\mathbb{R}^{3}$ minus a ball, is {\it asymptotically flat}. The matter is assumed compactly supported and no energy condition is required. A similar (though stronger) result applies to black holes too. This allows us to state a large generalisation of the uniqueness of the Schwarzschild solution not requiring asymptotic flatness. The Korotkin-Nicolai static black-hole shows that, for the given generalisation, no further flexibility in the hypothesis is possible.}
\end{spacing}
\end{minipage}

\n {\sc PACS}: 02.40.-k, 04.20.-q.

\n Keywords: General Relativity, static solutions, asymptotic.

\section{Introduction} 
Asymptotic flatness is the basic notion used in General Relativity (GR) to model systems that can be thought as ``isolated'' from the rest of the universe. It was used by Einstein himself at least in heuristic form and is now a standard piece of differential geometry and of gravitational and theoretical physics. 

The notion of asymptotic flatness is also epistemologically linked to the Newtonian theory of gravitation. 
\footnote{The following passage is made upon a text prepared by me to a highlight in CQG+.}In the 1916 manuscript {\it The Foundation of the Generalised Theory of Relativity}, Einstein addressed what he called {\it an epistemological defect} (but not mistake) of classical mechanics, whose origin he linked to E. Mach. He imagined two bodies, A and B, made of the same fluid material and sufficiently separated from each other that none of the properties of one could be attributed to the existence of the other. Observers at rest in one body see the other body rotating at a constant angular velocity, yet these same observers measure a perfect round surface in one case and an ellipsoid of rotation in the other case. He then asked: ``Why is this difference between the two bodies?". Necessarily, he continues, the answer cannot be found inside the system A+B only; It must lie in its exterior: the outer empty space. Einstein found that the source of the peculiar disparity was omitting that the empty space should also obey physical laws. These laws, which treat the parts A and B of the system A + B + {\sc Exterior Empty Space} on an equal footing, are the Einstein equations of GR. 
There is one point in Einstein's elegant conclusions that is left slightly inconclusive. It can be argued on the base of GR, that the absolute space of the 18th and 19th centuries was an inevitable concept, as ``corrections'' to the Newtonian gravity are simply too small. Though this is unquestionable, it can also be demanded to GR to explain too, why this ``background solution'', representing the {\sc Exterior Empty Space} of the system described earlier, is so distinguished in a theory that treats the geometry and the asymptotic of space, essentially as a variable. 

We find then that a problem of some theoretical importance is to analyse asymptotically flat (AF) solutions within the set of solutions of General Relativity and to find contexts in which they are indeed inevitable. Regardless of the ``aesthetic'' motivation just described, the study of the asymptotic of spacetimes is of course interesting in itself and can provide relevant information on structure of solutions to the Einstein equations. 

To give our result a framework, we redefine here {\it static isolated systems} in the simplest possible way without assuming asymptotic flatness at infinity. We prove then that these systems are necessarily AF. The definition of {\it static isolated system} is as follows. The region of the spacetime outside some set containing the sources should of the form
\be\label{STBOLD}
\stM=\mathbb{R}\times (\mathbb{R}^{3}\setminus \mathbb{B}^{3}),\qquad \stg=-N^{2}dt^{2}+g
\ee
where here $\mathbb{B}^{3}$ is the unit open ball in $\mathbb{R}^{3}$, $N>0$ is the lapse function (the norm of the static Killing field), and $g$ is a three-metric in $\mathbb{R}^{3}\setminus \mathbb{B}^{3}$. Moreover this spacetime region should be {\it geodesically complete until its boundary}, namely, spacetime geodesics (of any spacetime character) either end at its boundary or are defined for infinite parametric time. 

Admittedly, the topological condition, (which as we will se below is fundamental), is motivated mostly by historical considerations, although of course, to model a system like a neutron star, it is meaningless to make any other choice. On the other hand the geodesic completeness until the boundary of (\ref{STBOLD}) is the most basic condition that one can impose to ensure that the spacetime is, roughly speaking, ``endless''. From now on we will call it {\it geodesic completeness at infinity}; This terminology is justified by the following fact: geodesic completeness until the boundary holds iff every spacetime geodesic, whose projection into $\mathbb{R}^{3}\setminus \mathbb{B}^{3}$ leaves any compact set, is complete. 

In this setup we prove,

\begin{Theorem}\label{TII} Static isolated systems are asymptotically flat with Schwarzschildian fall off.
\end{Theorem}

This theorem is an expression of the remarkable consistence of General Relativity as a physical theory and shows the inevitability of asymptotic flatness in certain contexts.

\begin{figure}[h]
\centering
\includegraphics[width=7.5cm,height=6.5cm]{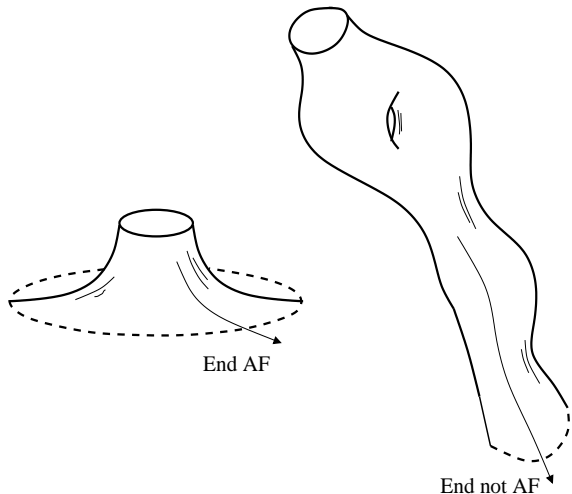}
\caption{Representation of an AF end and a non-AF end.}
\label{Figure1}
\end{figure}

To understand the importance and scope of the conditions defining static isolated systems, let us bring two purely relativistic examples into consideration. The first is the Schwarzschild black hole. It is a static vacuum solution with a topological-spherical hole, its curvature decays to zero at infinity, and the spacetime is geodesically complete at infinity.  Yet, (though not always properly emphasised), Schwarzschild it is not the only static vacuum black hole solution in 3+1 dimensions enjoying these attributes. The other solution we are referring to is the Korotkin-Nicolai static black hole \cite{94aperiodic}. It represents a topologically-spherical hole, not inside an open (infinite) three-ball $\mathbb{B}^{3}$ as in Schwarzschild, but inside an open (infinite) solid-torus $\mathbb{B}^{2}\times \mathbb{S}^{1}$. It is axially symmetric and has the asymptotic of a static Kasner \cite{94aperiodic} spacetime.  Its space is not simply connected; For this reason the horizon is prolate, as it feels the influence of itself along an axis of symmetry of finite length. The particular Kasner asymptotic is the simultaneous result of the presence of the hole on one side and of the non-trivial global topology on the other. Finite covers of the solution yield static spacetimes with a finite number of black holes in equilibrium. From the point of view of the General theory of Relativity, the Korotkin-Nicolai and the Schwarzschild solution are perfectly acceptable, still one is AF and the other is not. This shows that, in Theorem \ref{TII}, the topological assumption required for isolated systems cannot be removed. 
\vs

The proof of Theorem \ref{TII} is based on the results \cite{MR3233266}, \cite{MR3233267} where AF was proved under the extra hypothesis that (outside a compact set) $N$ is bounded from below away from zero\footnote{The definition of static isolated system in \cite{MR3233266} is the same as the one here, except that it includes the hypothesis that $N$ is bounded from below away from zero, see the remark inside the proof of Theorem \ref{TII}.}. This hypothesis was used only to guarantee that the conformal metric $N^{2}g$ is metrically complete, property that was used fundamentally. In a sense, all that we do in this article is to remove this undesired hypothesis on $N$ but for static solutions. We show that the completeness of $N^{2}g$ holds always in static isolated systems, as we defined them earlier. The techniques of this article do not apply directly to strictly stationary solutions (cf. Remark \ref{LIUSV}). The question of whether strictly stationary isolated systems are always AF is still open, though, (as shown in \cite{MR3233266}, \cite{MR3233267}), they are AF when the norm of the Killing field is bounded from below away from zero at infinity. 
\vs

Along the same lines as in Theorem \ref{TII}, we can generalise the celebrated uniqueness of the Schwarzschild solution (Israel \cite{Israel}\footnote{Israel breakthough in 1967, was the first uniqueness theorem for Schwarschild and required that $N$ could be chosen as a global radial coordinate.}, Robinson \cite{RobinsonII}, Bunting-Masood Um Alam \cite{MR876598}) to a uniqueness statement among an (a priori) much larger class of static solutions than those AF. Accordingly, we consider static solutions given by a vacuum static data $(\Sigma; g, N)$, i.e. with
\be
N Ric =\nabla\nabla N,\qquad \Delta N=0,
\ee
and with compact but not necessarily connected horizon $\partial \Sigma = \{N=0\}\neq \emptyset$. As earlier, the solutions are said to be {\it geodesically compete at infinity} if spacetime geodesics, of any spacetime character, either end at the horizon (i.e. the boundary) or are defined for infinite parametric time.  

The theorem is the following.

\begin{Theorem}\label{TI}
Let $(\sM;g,N)$ be the data set of a static vacuum spacetime with compact horizon and geodesically complete at infinity. Then, the spacetime is Schwarzschild iff a connected component of the complement of a compact set in $\Sigma$ is diffeomorphic to $\mathbb{R}^{3}\setminus \overline{\mathbb{B}^{3}}$.
\end{Theorem}

Observe that in this statement nothing is said about the (if any) other connected components of the complement of the compact set. In principle there could be many other unbounded connected components. That this cannot happen must be discerned after some analysis. This is in spirit similar to ``topological censorship'' - type of theorems as in \cite{MR1757638}, although our technique is different as we cannot rely on any given structure at infinity. 
\vs

In parallel to the discussion given at the beginning of the introduction, it is worth noting that Theorem \ref{TII} can be interpreted as a result on ``asymptotic uniqueness'', (here asymptotic flatness), and that, in this sense, it is a close relative of the uniqueness of the flat Minkowski spacetime among complete (simply connected) vacuum static spacetimes proved by M. T. Anderson in \cite{MR1806984}. Anderson's result is a direct consequence of a curvature decay that we will explain in Section \ref{BM}. We stress however that such decay is not nearly sufficient to deduce asymptotic flatness. The Korotkin-Nicolai solution satisfies this curvature decay and is not AF.

The rest of the article is roughly organised as follows. Sections \ref{BM}, \ref{BCP} and \ref{SPTG} deal with some important facts about the global structure of the vacuum static solutions. Section \ref{PROOF} contains the proofs of Theorems \ref{TII} and \ref{TI}. Proposition \ref{PI} shows the existence of a natural partition of static ends of the form $\mathbb{R}^{3}\setminus \mathbb{B}^{3}$. Proposition \ref{PII} then proves that the lapse $N$ can have only three types of behaviours at infinity and Proposition \ref{PIII} proves the completeness of $N^{2}g$ on the end. The proof of Theorems \ref{TII} and \ref{TI} are given afterwards.   

\begin{Acknowledgment} {\rm I am grateful to Marc Mars for interesting discussions on related topics.} 
\end{Acknowledgment}
\section{Background material.}\label{BMW}

A smooth Riemannian metric $g$ on a smooth connected manifold $\Sigma$ (with or without boundary, compact or not) induces the metric 
\be
\dist(p,q)= \inf\big\{\length(\gamma_{pq}):\gamma_{pq}\ \text{smooth curve joining $p$ to $q$}\big\}.
\ee
The space $(\Sigma;g)$ is said {\it metrically complete} if $(\Sigma; \dist)$ is complete. 
If $\Sigma$ has compact boundary then metric completeness is equivalent to the {\it geodesic completeness until the boundary} of $(\Sigma;g)$, (by Hopf-Rinow). On the other hand, geodesics in $(\Sigma;g)$ lift to geodesics perpendicular to the static Killing field in the associated spacetime. i.e. in 
\be
\stM=\mathbb{R}\times \Sigma,\qquad \stg=-N^{2}dt^{2}+g
\ee
Hence, if $\partial \Sigma$ is compact, geodesic completeness until the boundary of $(\stM;\stg)$ implies metric completeness of $(\Sigma;g)$. This is used in Proposition \ref{PIII}. 

Geodesic completeness until the boundary of $(\stM;\stg)$ is a basic assumption in the two main theorems in this article. However, regarding possible mathematical applications, it is important when possible to assume only the metric completeness of the data. We will make some remarks in this respect.
\vs

If $\partial \Sigma\neq \emptyset$, we define the {\it metric annulus} $\mathcal{A}(a,b)$ of radii $0<a<b$ by
\be
\mathcal{A}(a,b)=\big\{p\in \Sigma: a<\dist(p,\partial \Sigma)<b\big\}
\ee
where $\dist(p,\partial \Sigma)=\inf\{\dist(p,q): q\in \partial \Sigma\}$.    

\subsection{Anderson's curvature decay.}\label{BM}
{\it Anderson's curvature decay}\cite{MR1806984} is an important property of static solutions. It says that there is a universal constant $\const>0$ such that for any static data $(\Sigma; g,N)$ we have
\be\label{CURVDEC}
|Ric|(p)\leq \frac{\const}{\dist_{g}^{2}(p,\partial \Sigma)},\quad \text{and}\quad \bigg|\frac{\nabla N}{N}\bigg|^{2}(p)\leq \frac{\const}{\dist_{g}^{2}(p,\partial \Sigma)}
\ee
The optimal constant $\eta$ can be seen to be greater or equal than one, but it is not know if it is one. Upper bounds can be given but far from one. 

As an application of the curvature decay let us prove here a proposition that will be used in the proof of Theorem \ref{TI} to rule multiple ends when it is known that there is one that is AF. In the statement we use $\Sigma_{\delta}$ to denote the manifold resulting from removing from $\Sigma$ the tubular neighbourhood of $\partial \Sigma$ and radius $\delta$, i.e. $\Sigma_{\delta}=\Sigma\setminus \{p: \dist_{g}(p,\partial \Sigma)<\delta\}$. We assume that $\delta < \delta_{0}$ with $\delta_{0}$ small enough that $\partial \Sigma_{\delta}$ is always smooth.
\begin{Proposition}\label{PIV}
Let $(\Sigma;g,N)$ be a static vacuum initial data set with compact horizon ($\partial \Sigma=\{N=0\}\neq \emptyset$) and $(\Sigma;g)$ metrically compete. Then there is $0<\epsilon_{0}<1$ such that for every $\epsilon<\epsilon_{0}$ there is $\delta<\delta_{0}$ such that $(\Sigma_{\delta}; N^{-2\epsilon}g)$ is metrically complete and $\partial \Sigma_{\delta}$ is strictly convex (with respect to the inward normal).
\end{Proposition} 
\begin{proof}[\bf Proof.] Given $0<\epsilon<1$, the convexity of $\partial \Sigma_{\delta}$ for $\delta\leq \delta_{0}$ small enough is direct (and we leave it to the reader) as the factor $N^{-2\epsilon}$ ``blows up'' the boundary $\partial \Sigma$ uniformly, (observe however that, as $\epsilon<1$, $\partial \Sigma$ ``remains'' at a finite distance from the bulk of $\Sigma$).

So let us prove that, if we chose $\epsilon$ small enough, the space $(\Sigma_{\delta},N^{-2\epsilon}g)$ is metrically complete. As we are assuming $\delta< \delta_{0}$, it is enough to prove that, if $\epsilon$ is small enough, $(\Sigma_{\delta_{0}},N^{-2\epsilon}g)$ is metrically complete. We will do that below, the argument is thus independent of $\delta$.

It is enough to prove that, (if $\epsilon$ small enough), the following holds: for any sequence of points $p_{i}$ whose $g$-distance to $\partial \Sigma_{\delta_{0}}$ diverges, the $(N^{-2\epsilon}g)$-distance to $\partial \Sigma_{\delta_{0}}$ also diverges. Equivalently, it is enough to prove that for any sequence of curves $\gamma_{i}$ starting at $\partial \Sigma_{\delta_{0}}$ and ending at $p_{i}$ we have
\be
\int_{0}^{s_{i}}\frac{1}{N^{\epsilon}(\gamma_{i}(s))}ds\longrightarrow \infty
\ee
where $s$ is the $g$-arc length of $\gamma_{i}$ starting from $\partial \Sigma_{\delta_{0}}$. The curvature decay (\ref{CURVDEC}) implies right away the estimate, 
\be
N(p)\leq c(1+\dist_{g}(p,\partial \Sigma_{\delta_{0}}))^{\const}
\ee
for any $p\in \Sigma$ and where $\const>0$ is universal but $c$ depends on $(\Sigma,g)$ and $\delta_{0}$. As $\dist_{g}(\gamma_{i}(s),\partial \Sigma_{\delta_{0}})\leq s$, then we have
\be
N(\gamma_{i}(s))\leq c(1+s)^{\const}
\ee
Thus, if $\epsilon<1/\const$ then,
\begin{align}
\int_{0}^{s_{i}}\frac{1}{N^{\epsilon}(\gamma_{i}(s))}ds & \geq \int_{0}^{s_{i}} \frac{1}{c^{\epsilon}(1+s)^{\const\epsilon}}ds=\frac{1}{c^{\epsilon}(1-\epsilon \const)}\big((1+s_{i})^{1-\epsilon \const} -1\big)\\
& \geq \frac{1}{c^{\epsilon}(1-\epsilon \const)}\big(\big(1+\dist_{g}(p_{i},\partial \Sigma_{\delta_{0}})\big)^{1- \epsilon \const} -1\big)\longrightarrow \infty
\end{align}
as wished. 
\end{proof}
The importance of Proposition \ref{PIV} roots in that the Ricci curvature of the metric $\tilde{g}=N^{-2\epsilon}g$ has the expression\footnote{Use for this that if $\tilde{g}=e^{2\phi}g$ then $\tilde{Ric}=Ric-(\nabla\nabla \phi -\nabla\phi\nabla\phi)-(\Delta\phi+|\nabla\phi|^{2})g$ and that $\tilde{\nabla}_{i}V_{j}=\nabla_{i}V_{j}-(V_{j}\nabla_{i}\phi+ V_{i}\nabla_{j}\phi-(V^{k}\nabla_{k} \phi) g_{ij})$.}
\be
\tilde{Ric}=-\tilde{\nabla}\tilde{\nabla} f +\frac{1}{c}\tilde{\nabla}f\tilde{\nabla} f
\ee
where $f$ and $c$ depend on $\epsilon$ and are given by
\be
f=-(1+\epsilon)\ln N,\quad\text{and}\quad  \frac{1}{c}=\frac{(1-2\epsilon-\epsilon^{2})}{(1+\epsilon)^{2}} 
\ee
In particular, if $0<\epsilon<\sqrt{2}-1$ then $c>0$. This means that $c$-Bakry-Emery Ricci tensor $\tilde{Ric}_{f}^{c}$ given by
\be
\tilde{Ric}_{f}^{c}=\tilde{Ric}+\tilde{\nabla}\tilde{\nabla} f -\frac{1}{c}\tilde{\nabla}f\tilde{\nabla} f,
\ee
is zero. We will use this fundamentally later.

\subsection{The Ball Covering Property.}\label{BCP} 
As observed in \cite{MR1809792}, Liu's {\it ball covering property} holds for (metrically complete) static solutions $(\Sigma;g)$ with compact boundary. Namely, for any $0<a<b$ there is $r_{0}$ and $n_{0}$ such that for any $r\geq r_{0}$ there is a set of balls $\{B(p_{i},ar/2), p_{i}\in \overline{\mathcal{A}}(ar,br)$, $i=1,\ldots,n_{r}\leq n_{0}\}$, covering $\overline{\mathcal{A}}(ar,br)$. Here and below $\overline{\mathcal{A}}$ is the closure of $\mathcal{A}$.

As a direct corollary we have that, for any $0<a<b$ and $r\geq r_{0}$, as in the ball covering property, any two points in the same connected component of $\mathcal{A}(ar,br)$ can be joined by a curve of length less or equal than $n_{0}ar$ entirely contained in $\mathcal{A}(ar/3,3br)$.  

Let $\mathcal{A}_{c}(ar,br)$ be a connected component of $\mathcal{A}(ar,br)$. By the curvature decay (\ref{CURVDEC}) we have $|\nabla N/N|\leq 3\eta/ar$ all over $\mathcal{A}_{c}(ar/3,3br)$. Integrating this inequality along curves as in the previous paragraph we obtain 
\be\label{HARN}
\frac{\max\{N(p): p\in \overline{\mathcal{A}}_{c}(ar,br)\}}{\min\{N(p):p\in \overline{\mathcal{A}}_{c}(ar,br)\}}\leq C(a,b)
\ee
This is a type of Harnack inequality for $N$ and is fundamental.

\begin{Remark}\label{LIUSV} It is not known at the moment if a similar ball covering property holds for strictly stationary solutions. This is a main obstacle to extend Theorem \ref{TII} to stationary isolated systems.
\end{Remark}

\subsection{Spacetime geodesics in static spacetimes.}\label{SPTG}
Let $(\Sigma; g, N)$ be a static vacuum data and let $(\stM,\stg)$ be its associated spacetime. We recall here a useful way to describe spacetime geodesics $\stgeo(\tau)$ in terms of certain metrics conformal to $g$ in $\Sigma$. This goes back at least to the work of H. Weyl \cite{Weyl} from 1917. 

Let $\sgeo=\Pi(\stgeo)$ be the projection of $\stgeo$ into $\Sigma$. Then it is direct to see that $\sgeo$ satisfies the equation
\be
\nabla_{\sgeo'}\sgeo'=a^{2}\frac{\nabla N}{N^{3}}
\ee
where $\gamma'=d\gamma/d\tau$ and where $a$ is the constant $a=\stg(\stgeo',\partial_{t})$. Moreover we have
\be
|\sgeo'|^{2}=\varepsilon+\frac{a^{2}}{N^{2}}
\ee
where the norm on the l.h.s is with respect to $g$ and where $\varepsilon=-1,0,1$ according to the character type of the geodesic.

Then define $e^{2\phi}$ by
\be
e^{2\phi}=(\varepsilon+\frac{a^{2}}{N^{2}})
\ee
wherever the right hand side is positive (this includes the projection of the geodesic). Finally consider the conformal metrics 
\be
\hat{\sg}=e^{2\phi}g,\qquad \check{\sg}=e^{-2\phi}g. 
\ee
and denote by $ds$, $d\hat{s}=e^{f}ds$, and $d\check{s}=e^{-f}ds$ the elements of length of $\sgeo$ with respect to $g, \hat{g}$ and $\check{g}$ respectively. 

In this setup we have the following characterisation: {\it If $\stgeo(\tau)$ is a spacetime geodesic then $\sgeo(\hat{s})$ is a geodesic of $\hat{g}$ and $d\tau=d\check{s}$. Conversely if $\sgeo(\hat{s})$ is a geodesic of $\hat{g}$ then the curve
\be
\stgeo(\check{s})=(\int^{\check{s}}\frac{a}{N(\sgeo(\check{s}'))}d\check{s}',\sgeo(\check{s}))\subset \mathbb{R}\times \sM = \stM
\ee
is a spacetime geodesic with $\stg(\stgeo',\stgeo')=\varepsilon$, hence with $\tau=\check{s}$.}

Two points are particularly important about this characterisation of spacetime geodesics, (i) spacetime geodesics can be constructed out of the projected curves which in turn can be easily found through length-minimisation, (ii) as the affine parameter of spacetime geodesics is the $\check{g}$-arc length of the projected curve, a way is opened to link spacetime geodesic completeness at infinity to the metric completeness of $\check{g}=N^{2}g$. We will exploit these two observations during the proof of Proposition \ref{PIII}. We will use only the characterisation of null geodesics, i.e. $\epsilon=0$, although other types of geodesics can be useful in similar contexts. 

\section{The proofs}\label{PROOF}

Every, smooth, connected, compact, boundaryless and orientable surface $F$ embedded in $\mathbb{R}^{3}$ divides $\mathbb{R}^{3}$ in two connected components. Below we will work with such surfaces $F$ embedded in $\mathbb{R}^{3}\setminus \overline{\mathbb{B}^{3}}$ and will denote by $M(F)$ to the closure of the bounded connected component of $(\mathbb{R}^{3}\setminus \mathbb{B}^{3})\setminus F$.  
Two facts are direct to check. First, for any disjoint $F_{1}$ and $F_{2}$ such that $\partial \mathbb{B}^{3}\subset M(F_{i})$ for $i=1,2$, then either $F_{1}\subset M^{\circ}(F_{2})$ or $F_{2}\subset M^{\circ}(F_{1})$, (here $^{\circ}=\text{Interior}$). Second, if a set $\{F_{i}, i=1,\ldots,n\geq 1\}$ of such surfaces is such that $\partial \mathbb{R}^{3}$ belongs to a bounded component of $\Sigma\setminus \bigcup_{i=1}^{i=n}F_{i}$ then there is at least one $F_{i}$ such that $\partial \mathbb{B}^{3}\subset M(F_{i})$. We will use these facts in the proof of the following proposition.

\begin{Proposition}\label{PI} Let $(\Sigma;g,N)$ be a metrically complete vacuum static data set with $\Sigma\approx \mathbb{R}^{3}\setminus \mathbb{B}^{3}$. Then, there is a set of (smooth, connected, compact, boundaryless and orientable) surfaces $\{S_{j}; j=0,1,2,3,\ldots\}$, such that the following holds for every $j$,
\begin{enumerate}
\item $S_{j}$ is embedded in $\An(2^{1+2j},2^{2+2j})$,
\item $\partial \Sigma\subset M(S_{j})$,
\item $M(S_{j})\subset M(S_{j+1})$,
\end{enumerate}
\end{Proposition}
\vs

The surfaces $S_{i}$ will be used only as references inside the manifold $\Sigma$, their geometries play no role. Observe that $\Sigma\setminus M(S_{k})=\bigcup_{j=k}^{j=\infty} M(S_{j+1})\setminus M(S_{j})$ with the union disjoint and that $S_{j+1}\cup S_{j}=\partial (M(S_{j+1})\setminus M^{\circ}(S_{j}))$. This last observation will be used when we apply the maximum principle to $N$ on $M(S_{j+1})\setminus M^{\circ}(S_{j})$.

\begin{proof}[\bf Proof.] 
In the argument that follows we treat $\Sigma$ and $\mathbb{R}^{3}\setminus \mathbb{B}^{3}$ indistinctly. The construction of the surfaces $S_{j},\ j=0,1,2,\ldots$ is as follows. Let $f:\Sigma\rightarrow [0,\infty)$ be a (any) smooth function such that $f\equiv 1$ on $\{p:\dist(p,\partial \Sigma)\leq 2^{1+2j}\}$ and $f\equiv 0$ on $\{p: \dist(p,\partial \Sigma)\geq 2^{2+2j}\}$. Let $x$ be any regular value of $f$ in $(0,1)$. Then we can write $f^{-1}(x)=F_{1}\cup\ldots\cup F_{n}$ where each $F_{i}$ is a (connected, compact, boundaryless and orientable) surface embedded in $\mathcal{A}(2^{1+2j},2^{2+2j})$. 
Now, as $\Sigma$ is the disjoint union of the sets $f^{-1}((x,\infty))$, $f^{-1}(x)=\bigcup_{i=1}^{i=\infty}F_{i}$ and $f^{-1}((-\infty,x))$, and as $\{p: \dist(p,\partial \Sigma)\geq 2^{2+2j}\}\subset f^{-1}((-\infty,x))$ we conclude that $\partial \Sigma$, which lies inside $f^{-1}((x,\infty))$, must belong to a bounded component of $\Sigma\setminus \bigcup_{i=1}^{i=n}F_{i}$. Hence $\partial \Sigma\subset M(F_{i*})$ for some $F_{i*}$, (see the beginning of this section). We set $S_{j}=F_{i*}$.

We verify now that the surfaces $S_{j}$ satisfy the properties 1-3. By construction the $S_{j}$'s satisfy already 1 and 2. Now, either $M(S_{j})\subset M^{\circ}(S_{j+1})$ or $M(S_{j+1})\subset M^{\circ}(S_{j})$. If $M(S_{j+1})\subset M^{\circ}(S_{j})$ then $S_{j+1}\subset \{p:\dist(p,\partial \Sigma)<2^{2+2j}\}$ which is impossible because $S_{j+1}\subset \mathcal{A}(2^{3+2j},2^{4+2j})$. Thus, $M(S_{j})\subset M^{\circ}(S_{j+1})$, showing property 3.
\end{proof}

We claim that, for any $j\geq 0$, the surfaces $S_{j+1}$ and $S_{j}$ lie in the same connected component of the annuli $\mathcal{A}(2^{1+2j},2^{4+2j})$. To see this, consider a ray $\gamma(s)$, $s\geq 0$, starting at $\partial\Sigma$ at $s=0$, (i.e. $\dist(\gamma(s),\partial \Sigma)=s$ for all $s\geq 0$; $s$ is arc-length). Let $s_{j}$ be the last time that $\gamma(s)\in S_{j}$ and let $s_{j+1}$ be the first time that $\gamma(s)\in S_{j+1}$. Then, $s_{j}\geq 2^{1+2j}$ because $S_{j}\subset \mathcal{A}(2^{1+2j},2^{2+2j})$ and $s_{j+1}\leq 2^{4+2j}$ because $S_{j+1}\subset \mathcal{A}(2^{3+2j},2^{4+2j})$. Hence the arc $\{\gamma(s): s\in [s_{j},s_{j+1}]\}$ must lie inside $\mathcal{A}(2^{1+2j},2^{4+2j})$ because $\dist(\gamma(s),\partial \Sigma)=s$ for all $s$. We conclude then that $S_{j}$ and $S_{j+1}$ must lie in the same connected component of $\mathcal{A}(2^{1+2j},2^{4+2j})$.       

This claim and Proposition \ref{PI} will be used in the proof of the following proposition. 

\begin{Proposition}\label{PII} Let $(\Sigma;g,N)$ be a metrically complete vacuum static data set with $\Sigma\approx \mathbb{R}^{3}\setminus \mathbb{B}^{3}$ and $N>0$. Then, one of the following holds,
\begin{enumerate}
\item $N$ converges uniformly to zero over the end of $\Sigma$,
\item $N$ converges uniformly to infinity over the end of $\Sigma$,
\item $C_{1}<N<C_{2}$ for constants $0<C_{1}<C_{2}<\infty$.
\end{enumerate}
\end{Proposition} 

\begin{proof}[\bf Proof.] To shorten notation we will write $\max\{N;\Omega\}:=\max\{N(p):p\in \Omega\}$ where $\Omega$ are compact sets (same notation for $\min\{N;\Omega\}$).

Suppose that there is a divergent sequence $p_{i}$ for which $N(p_{i})\rightarrow 0$ as $i\rightarrow \infty$. We claim that, in this case, $N$ tends uniformly to zero over the end. 

For every $i$ let $j_{i}$ be such that $p_{i}\in M(S_{j_{i}})\setminus M^{\circ}(S_{j_{i}-1})$. Suppose first that
\be\label{DECA}
\max\{N;S_{j_{i}}\}\rightarrow 0
\ee
Then, for any $i'>i$ the maximum principle gives
\be
\max\{N; M(S_{j_{i'}})\setminus M^{\circ}(S_{j_{i}})\}\leq \max\big\{\max\{N; S_{j_{i'}}\}, \max\{N; S_{j_{i}}\}\big\}
\ee
Letting $i'\rightarrow \infty$ and using (\ref{DECA}) we obtain
\be
\sup\{N(p):p\in \Sigma\setminus M^{\circ}(S_{j_{i}})\}\leq \max\{N;S_{j_{i}}\}
\ee
where the r.h.s tends to zero as $i$ tends to infinity. This proves that $N$ tends uniformly to zero as claimed. 

To prove (\ref{DECA}) we recall first that $S_{j_{i}}$ and $S_{j_{i}-1}$ lie in the same connected component of $\mathcal{A}(2^{2j-1},2^{2j+2})$. Therefore, as commented in Section \ref{BCP}, we have
\be\label{A}
\max\{N;S_{j_{i}}\}\leq c\min\{N;S_{j_{i}}\cup S_{j_{i}-1}\}
\ee
where the constant $c$ is independent of $i$. On the other hand, by the maximum principle we have
\be\label{B}
\min\{N;S_{j_{i}}\cup S_{j_{i}-1}\}\leq \min\{N;M(S_{j_{i}})\setminus M^{\circ}(S_{j_{i}-1})\}\leq N(p_{i})
\ee
Combining (\ref{A}) and (\ref{B}) we obtain
\be\label{C}
\max\{N;S_{j_{i}}\}\leq N(p_{i})
\ee
where the r.h.s tends to zero. This implies (\ref{DECA}) as desired.

In the same manner one proves that if there is a divergent sequence $p_{i}$ such that $N(p_{i})\rightarrow \infty$ as $i\rightarrow \infty$ then $N$ tends uniformly to infinity over the end. 

If none of the situations considered above occurs then $0<C_{1}<N<C_{2}$ for constants $C_{1}, C_{2}$. 
\end{proof} 

To show asymptotic flatness for isolated systems using \cite{MR3233266}, \cite{MR3233267}, we need only to prove the completeness of $N^{2}g$ using that the static spacetime is geodesically complete at infinity. This is done in the next proposition.

\begin{Proposition}\label{PIII} Let $(\Sigma;g,N)$ be a static vacuum data set, with $\Sigma\approx \mathbb{R}^{3}\setminus \mathbb{B}^{3}$ and $N>0$ on $\Sigma$. Assume that the associated spacetime 
\be
\stM=\mathbb{R}\times \Sigma,\qquad \stg=-N^{2}dt^{2}+g
\ee
is geodesically complete at infinity. Then the space $(\Sigma; N^{2}g)$ is metrically complete.
\end{Proposition}

\begin{proof}[\bf Proof.]
The proof is by contradiction. So let us assume that $(\Sigma; N^{2}g)$ is not metrically complete. We will explain later how this contradicts  the {\it geodesic completeness at infinity}. During the proof we use the same notation as in Proposition \ref{PII}.
We will also use, as was explained in Section \ref{BMW}, that under the hypothesis of the proposition, the space $(\Sigma;g)$ is metrically complete. 

We begin by proving that 
\be\label{TOPR}
\sum_{j=1}^{j=\infty} \max\{N;S_{j}\}2^{2j}<\infty
\ee
Let $\beta:[s_{j},s_{j+1}]\rightarrow M(S_{j+1})\setminus M^{\circ}(S_{j})$ be any curve with $\beta(s_{j})\in S_{j}$ and $\beta(s_{j+1})\in S_{j+1}$. We claim that then
\be\label{NINT}
\int_{s_{j}}^{s_{j+1}}N(\beta(s))ds\geq c_{1}\max\{N;S_{j}\}2^{2j} 
\ee
where the constant $c_{1}$ is independent of $j$. To see this write
\be
\int_{s_{j}}^{s_{j+1}}N(\beta(s))ds\geq \min\{N;M(S_{j+1})\setminus M^{\circ}(S_{j})\} \length(\beta)
%\geq \min\{N;S_{j}\}\length(\beta)
\ee
and note that, 
\begin{enumerate}
\item $\length(\beta)\geq (2^{3+2j}- 2^{1+2j})=6\,2^{2j}$, because $S_{j}\subset \mathcal{A}(2^{1+2j},2^{2+2j})$ and $S_{j+1}\subset \mathcal{A}(2^{3+2j},2^{4+2j})$, and,
\item $\min\{N;M(S_{j+1})\setminus M^{\circ}(S_{j})\}\geq \max\{N;S_{j}\}$, because 
\be
\min\{N;M(S_{j+1})\setminus M^{\circ}(S_{j})\}\geq \min\{N;S_{j+1}\cup S_{j}\}
\ee
by the maximum principle, and because 
\be
\min\{N;S_{j+1}\cup S_{j}\}\geq c_{2}\max\{N;S_{j}\},
\ee
where $c_{2}$ is independent of $j$, by what was explained in Section \ref{BCP}, (see also the remark after the proof of Prop. \ref{PI}). 
\end{enumerate}
The formula (\ref{NINT}) is then obtained making $c_{1}=6c_{2}$. 

Now, if $(\Sigma; N^{2}g)$ is not metrically complete, then one can find a sequence of points $p_{i}$, with $\dist_{g}(p_{i},\partial \Sigma)\rightarrow \infty$ but with $\dist_{N^{2}g}(p_{i},\partial \Sigma)$ uniformly bounded. From the definition of $\dist$, this implies that there is a sequence of curves $\alpha_{i}(s);\ s\in [0,s_{i}]$ starting at $\partial \Sigma$ and ending at $p_{i}$, for which
\be
\int_{s=0}^{s=s_{i}}N(\alpha(s))ds\leq K<\infty
\ee
where $K$ is independent of $j$. For every $i$ let $j_{i}$ be the greatest $j$ such that $p_{i}\notin M(S_{j})$. Then, for every $j\leq j_{i}-1$ one can find an interval $[s_{j,i},s_{j+1,i}]$ such that the curve $\beta_{j}$ defined by $\beta_{j}(s)=\alpha_{i}(s)$, $s\in [s_{j,i},s_{j+1,i}]$, has range in $M(S_{j+1})\setminus M^{\circ}(S_{j})$ and moreover with $\beta_{j}(s_{j,i})\in S_{j}$ and $\beta_{j}(s_{j+1,i})\in S_{j+1}$. Using 
(\ref{NINT}) we write
\be
K\geq \int_{s=0}^{s=s_{i}}N(\alpha)ds\geq \sum_{j=1}^{j=j_{i}-1} \int_{s_{j,i}}^{s_{j+1,i}}N(\beta_{j})ds\geq \sum_{j=1}^{j=j_{i}-1}c_{1}\max\{N;S_{j}\} 2^{2j}
\ee  
Taking the limit $i\rightarrow \infty$ gives (\ref{TOPR}) as wished.

We proceed now with the proof. By Proposition \ref{PII} we know that $N$ must go uniformly to zero at infinity otherwise $N$ would be bounded from below away from zero and the metric $N^{2}g$ would be automatically complete. If $N\rightarrow 0$ uniformly at infinity, then $(\Sigma;N^{-2}g$) is metrically complete. 

As was explained in Section \ref{SPTG}, null-spacetime geodesics project into $(N^{-2}g)$-geodesics and the affine parameter is the $(N^{2}g)$-arc length. We will see below that if $(\Sigma; N^{2}g)$ is not metrically complete then there is an infinite $(N^{-2}g)$-geodesic whose $(N^{2}g)$-length is finite. This would be against the hypothesis that the spacetime is geodesically complete at infinity and the proof will be finished. 

Let $\Gamma(s)$, $s\geq 0$ be a ray for the metric to $N^{-2}g$ and starting at $\partial \Sigma$. For each $j\geq 1$ let $s_{j}$ be the last time that $\Gamma(s)\in S_{j}$. Let $\Gamma_{j}$ be the restriction of $\Gamma$ to $[s_{j},s_{j+1}]$. Then $\Gamma_{j}\subset (\Sigma\setminus M^{\circ}(S_{j}))$ and $\Gamma$ is the concatenation of the curves $\Gamma_{j}$, $j\geq 1$. Now,
\be\label{PRO}
\int_{s=s_{1}}^{s=\infty}N(\Gamma(s))ds=\sum_{j=1}^{j=\infty} \int_{s_{j}}^{s_{j+1}}N(\Gamma_{j}(s))ds\leq \sum_{j=1}^{j=\infty}
\max\{N;S_{j}\}\length(\Gamma_{j})
\ee
where to obtain the inequality we use that,  
\be
\sup\{N(\Gamma_{j}(s)): s\in [s_{j},s_{j+1}]\}\leq \sup\{N(p):p\in \Sigma\setminus M^{\circ}(S_{j})\}\leq \max\{N;S_{j}\}.
\ee 
which is obtained from the inclusion $\Gamma_{j}\subset(\Sigma\setminus M^{\circ}(S_{j}))$ (for the first inequality), and from the maximum principle (for the second). Thus, if we prove that for a constant $c_{3}$ independent of $j$ we have 
\be\label{GGG}
\length(\Gamma_{j})\leq c_{3}2^{2j}
\ee
then we can use (\ref{TOPR}) in conjunction to (\ref{PRO}) to conclude that 
\be
\int N(\Gamma(s))ds<\infty
\ee
which would imply that there is an incomplete null geodesic in the spacetime.

Let us prove then the inequality (\ref{GGG}). We will play with the fact that $\Gamma$ is a ray for $N^{-2}g$.

First note
\be\label{LETFIN}
\int_{s_{j}}^{s_{j+1}}\frac{1}{N(\Gamma_{j}(s))}ds\geq \frac{\length(\Gamma_{j})}{\max\{N;\Gamma_{j}\}}\geq \frac{\length(\Gamma_{j})}{\max\{N;S_{j}\}}
\ee
where the second inequality is obtained from the inclusion $\Gamma_{j}\subset \Sigma\setminus M^{\circ}(S_{j})$ and because $\max\{N;\Sigma\setminus M(S_{j})\}\leq \max\{N;S_{j}\}$ by the maximum principle.

Then recall from the discussion after Proposition \ref{PI}, that $S_{j}$ and $S_{j+1}$ lie in the same connected component $\mathcal{A}_{c}(2^{1+2j},2^{4+2j})$ of $\mathcal{A}(2^{1+2j},2^{4+2j})$. Hence, $\Gamma(s_{j})(\in S_{j})$ and $\Gamma(s_{j+1})(\in S_{j+1})$ lie also in $\mathcal{A}_{c}(2^{1+2j},2^{4+2j})$. Then, as in Section \ref{BCP}, we can joint $\Gamma(s_{j})$ to $\Gamma(s_{j+1})$ through a curve $\Gamma'_{j}$ of length less or equal than $c2^{2j}$, ($c$ is a constant independent of $j$), entirely contained in a connected component $\mathcal{A}_{c}(2^{1+2j}/3,32^{4+2j})$ of $\mathcal{A}(2^{1+2j}/3,3\,2^{4+2j})$. 
This curve $\Gamma_{j}'$ must have $(N^{-2}g$)-length greater or equal than the $(N^{-2}g$)-length of $\Gamma_{j}$ because $\Gamma_{j}$, (being a ray), minimises the  $(N^{-2}g$)-length between any two of its points. Thus we can the write
\be
\int_{s_{j}}^{s_{j+1}}\frac{1}{N(\Gamma(s))}ds\leq \int_{s'_{j}}^{s'_{j+1}}\frac{1}{N(\Gamma_{j}'(s'))}ds'\leq \frac{c2^{2j}}{\min\{N;\overline{\mathcal{A}}_{c}(2^{1+2j}/3,3\,2^{4+2j})\}}
\ee
Together with (\ref{LETFIN}) we obtain
\be
\length(\Gamma_{j})\leq c\bigg[\frac{\max\{N;\Gamma_{j}\}}{\min\{N;\overline{\mathcal{A}}_{c}(2^{1+2j}/3,32^{4+2j})\}}\bigg]\ 2^{2j}
\ee
But from (\ref{HARN}) we have
\be
\frac{\max\{N;\Gamma_{j}\}}{\min\{N;\overline{\mathcal{A}}_{c}(2^{1+2j}/3,32^{4+2j})\}}\leq \frac{\max\{N;\overline{\mathcal{A}}_{c}(2^{1+2j}/3,32^{4+2j})\}}{\min\{N;\overline{\mathcal{A}}_{c}(2^{1+2j}/3,32^{4+2j})\}}\leq c'
\ee
where $c'$ is independent of $j$. Thus, (\ref{GGG}) follows.
\end{proof}

\begin{proof}[\bf Proof of Theorem \ref{TII}.] From the same definition of static isolated system, we know that the spacetime outside a set  (invariant under the Killing field) is 
\be
\stM=\mathbb{R}\times (\mathbb{R}^{3}\setminus \mathbb{B}^{3}),\qquad \stg=-N^{2}dt^{2}+g
\ee
which is described by the data $(\mathbb{R}^{3}\setminus \mathbb{B}^{3}; g,N)$. As the spacetime is geodesically complete at infinity  we can use Proposition \ref{PIII} to deduce that the metric $N^{2}g$ is complete on $\mathbb{R}^{3}\setminus \mathbb{B}^{3}$. Theorem 1.3 in \cite{MR3233266} then apples and asymptotic flatness follows. 

(Remark: The notion of Isolated System used in \cite{MR3233266} is the same as in this paper but with the extra assumption that $N$ is bounded from below away from zero outside a compact set. As commented in \cite{MR3233266}, Theorem 1.3 still holds if this hypothesis on $N$ is replaced by the metric completeness of $N^{2}g$.) 
\end{proof}

\begin{Remark} If the matter model, (which is always assumed compactly supported), satisfies the weak energy condition then the conclusions of Theorem \ref{TII} can be seen to follow only from the metric completeness of the static data. The geodesic completeness at infinity is unnecessary.
\end{Remark}

We can now prove Theorem \ref{TI}.

\begin{proof}[\bf Proof of Theorem \ref{TI}.]
Suppose that a connected component of the complement of a compact set in $\Sigma$ is diffeomorphic to $\mathbb{R}^{3}$ minus a closed ball. Then, as in the proof of Theorem \ref{TII}, this component has to be an AF end of $\Sigma$. If we prove that $\Sigma$ has only one end, then the Main Theorem  in \cite{MR1201655} shows that $\Sigma$ is diffeomoprhic to $\mathbb{R}^{3}$ minus a finite set of open balls. The Israel - Robinson - Bunting - Masood-ul-Alam uniqueness Theorem then applies and the solution is Schwarzschild.  Let us prove then that $\Sigma$ must have only one end.

We will proceed by contradiction. Assume then that $\Sigma$ has more than one end. From now on we work in a space $(\Sigma_{\delta},N^{-2\epsilon}g)$ as in Proposition \ref{PIV} but with $\epsilon<\sqrt{2}-1$. 

The end that was AF (and had Schwarzschildian fall off) for $g$ is also AF for $N^{-2\epsilon}g$. On this end consider large (``almost round'') embedded spheres $S$. On them we have $|\nabla N|_{N^{2-\epsilon}g}\lesssim 1/{\rm area}(S)$ while for the mean curvature $\theta_{S}$, (with respect to the outward unit normal $n$), we have $\theta_{S}\approx 2\sqrt{4\pi/{\rm area}(S)}$. Hence one can clearly take an embedded sphere $S$ sufficiently far away that 
\be\label{LHS}
\theta_{S}-(1+\epsilon)\frac{n(N)}{N}>0
\ee
at every point of $S$. We work with such $S$ below. The particular combination (\ref{LHS}) will be relevant. The sphere $S$ divides $\Sigma_{\delta}$ in two connected components. Denote by $\Sigma'_{\delta}$ the closure of the connected component of $\Sigma_{\delta}\setminus S$ containing $\partial \Sigma$. We have $\partial \Sigma'_{\delta}=\partial \Sigma \cup S$ and, more importantly, $\Sigma'_{\delta}$ contains at least one more end. Since $\partial \Sigma_{\delta}$ is strictly convex, we can construct a geodesic ray $\gamma(s)$, $s\geq 0$, in $\Sigma'_{\delta}\setminus \partial \Sigma$ and with the following properties, 
\begin{enumerate}
\item $\gamma(s)$ starts at $S$ and perpendicularly to it,
\item $\gamma(s)$ diverges through and end in $\Sigma'_{\delta}$ as $s\rightarrow \infty$,
\item $\dist_{N^{-2\epsilon}g}(\gamma(s),S)=s$ for all $s\geq 0$.     
\end{enumerate}
These properties imply that the expansion $\theta(s)$, along the geodesic $\gamma(s)$, of the congruence of geodesics emanating perpendicularly to $S$, must remain finite for all $s$ (i.e. $\theta(s)>-\infty$ for all $s\geq 0$). If not then there is a focal point on $\gamma$ after which property 3 fails. We will prove now that indeed $\theta(s)=-\infty$ for some $s>0$, thus reaching a contradiction.

Let 
\be
m(s)=\theta(s)+(1+\epsilon)\frac{N'(s)}{N(s)}
\ee
where $N(s)=N(\gamma(s))$ and $N'(s)=d N(\gamma(s))/ds$. At $s=0$, $m$ is equal to minus the left hand side of (\ref{LHS}), therefore negative (note that $\gamma'(0)=-n$). On the other hand, as we explained in Section \ref{BM}, if $\epsilon<\sqrt{2}-1$ then the Bakry-Emery Ricci tensor 
\be
Ric_{f}^{c}=Ric +\nabla\nabla f-\frac{1}{c}\nabla f\nabla f
\ee
is zero, where $f=(1+\epsilon)\ln N$ and $1/c=(1-2\epsilon-\epsilon^{2})/(1+\epsilon)^{2}$. Now, it is shown in \cite{MR2577473} (Appendix A) that $m(s)$ satisfies the differential inequality
\be
m'\leq -\frac{m^{2}}{2+c}
\ee
Thus, if $m(0)<0$ then there is $s'>0$ such that $m(s')=-\infty$. But as $N'(s)/N(s)$ is finite for all $s$ then we must have $\theta(s')=-\infty$.         
\end{proof}  

\begin{Remark} If the complement of a compact set in $\Sigma$ is diffeomorphic to $\mathbb{R}^{3}\setminus \overline{\mathbb{B}^{3}}$ and $(\Sigma;g)$ is metrically complete, then the solution is Schwarzschild too, (i.e. the geodesic completeness of the spacetime at infinity is unnecessary). To see this observe first that $N$ cannot go uniformly to zero on the end of $\Sigma$ because this would violate the maximum principle ($N$ is harmonic and is zero only on $\partial \Sigma$). By Proposition \ref{PIII} $N$ is then bounded away from zero on the end and asymptotic flatness follows. 
\end{Remark}

\begin{Remark} It is easy to show that Propositions \ref{PI}, \ref{PII} and \ref{PIII} hold true when $\Sigma \approx \mathcal{S}\times \mathbb{R}^{+}$ with $\mathcal{S}$ a compact two-surface of arbitrary genus, (Proposition \ref{PI} corresponds to $\mathcal{S}=\mathbb{S}^{2}$). This could be of interest in further studies.
\end{Remark}

\bibliographystyle{plain}
\bibliography{Master}

\end{document}